\documentclass[runningheads,a4paper]{llncs}
\usepackage{graphicx}
\graphicspath{{Figures/}}
\usepackage{chngcntr}
\usepackage{amssymb}
\usepackage{booktabs}
\usepackage{lineno}

\usepackage{amsmath}
\usepackage{amsfonts}

\makeatletter
\pdfoutput=1

\newcommand{\deleted}[1]{}
\usepackage {color}
\usepackage{lineno}
\usepackage{algo}
\usepackage{verbatim}

\newtheorem{observation}{Observation}

\newcommand*{\level}{\ell}

\newcommand{\Plane}{\mathbb{R}^2}
\newcommand{\HVD}[1]{\ensuremath{\textup{\textsf{HVD}}(#1)}}
\newcommand{\FVD}[1]{\ensuremath{\textup{\textsf{FVD}}(#1)}}
\newcommand{\hreg}[2]{\ensuremath{\textup{\textsf{hreg}}_{#1}(#2)}}
\newcommand{\freg}[2]{\ensuremath{\textup{\textsf{freg}}_{#1}(#2)}}
\newcommand{\fskel}[1]{\ensuremath{{\mathcal{T}}(#1)}}
\newcommand{\cl}[1]{\ensuremath{\overline{#1}}}
 \let\Ranges\ranges
 \let\conf\confl
\newcommand{\df}[1]{\ensuremath{\textup{\textsf{d}}_\textup{\textsf{f}}}(#1)}
\newcommand{\bh}[1]{\ensuremath{\textup{\textsf{b}}_\textup{\textsf{h}}}(#1)}

\newcommand{\act}[2]{\mathcal{T}_a(#2,#1)} 

\newcommand{\bd}{\partial}
\newcommand{\f}{\ensuremath{f}}

\newcommand{\vlist}[1]{\ensuremath{V(#1)}}

\newcommand*{\nclus}{k}
\newcommand{\AbsFamily}{F} \let\absFamily\AbsFamily
\newcommand{\curClus}{C}
\newcommand{\InpSet}{F}
\newcommand{\aClus}{C}
\newcommand{\pointFromClus}{c}
\newcommand{\aPoint}{t}

\newcommand*{\Hist}[1]{\mathcal{H}(#1)}

\newcommand{\M}{m}

\newcommand{\etal}{{~et~al}.\xspace}

\newcommand{\evanthia}[2][says]{
\@ifundefined{showmessages}{\relax}
{** \textsc{evanthia #1:} \textcolor{red}{\textsl{#2}} **}}

\newcommand{\elena}[2][says]{
\@ifundefined{showmessages}{\relax}{
** \textsc{elena #1:} \textcolor{magenta}{\textsl{#2}} **}
}

\makeatletter
\renewcommand*{\@fnsymbol}[1]{\ensuremath{\ifcase#1\or *\or \dagger\or \ddagger\or
    \mathsection\or \mathparagraph\or \|\or **\or \dagger\dagger
    \or \ddagger\ddagger \else\@ctrerr\fi}}
\makeatother

\begin{document}
\sloppy
\mainmatter

\title{Randomized Incremental Construction for the Hausdorff Voronoi
  Diagram revisited and extended\thanks{
Research 
supported in part by
    the Swiss National Science Foundation, projects
    SNF 20GG21-134355 (ESF EUROCORES EuroGIGA/VORONOI) and  SNF 200021E-154387.
E. A. was also supported partially by F.R.S.-FNRS and SNF grant
P2TIP2-168563 under the SNF Early PostDoc Mobility program.}\thanks{This is a pre-print of an article published in Journal of Combinatorial Optimization. 
The final authenticated version is available online at: https://doi.org/10.1007/s10878-018-0347-x}}
\titlerunning{RIC for HVD revisited and extended}

\author{Elena Arseneva\inst{1}\thanks{Research performed mainly 
    while at the Universit\`a della Svizzera
  italiana (USI).} 
\and
        Evanthia Papadopoulou\inst{2}}

\institute{
St. Petersburg State University (SPbU), Russia,
{\tt ea.arseneva@gmail.com}
\and
Faculty of Informatics,  Universit\`a della Svizzera
  italiana (USI), Lugano, Switzerland,
 {\tt evanthia.papadopoulou@usi.ch}}

\maketitle

\begin{abstract}
The Hausdorff Voronoi diagram of clusters of points in the
plane is a generalization of Voronoi diagrams based on the
Hausdorff distance function. 
Its combinatorial complexity is $O(n+\M)$, where $n$ is the total 
number of points and $\M$
is the number of \emph{crossings} between 
the input clusters ($\M=O(n^2)$);
the number of clusters is $k$.
We present efficient  
algorithms to construct  this diagram following 
the randomized incremental construction (RIC) framework 
[Clarkson\etal~89,~93].
Our algorithm for \emph{non-crossing} clusters ($\M=0$) 
runs  in
expected $O(n\log{n} + k\log n \log k)$ time and deterministic
 $O(n)$ space.
The algorithm for arbitrary clusters runs in 
expected $O((\M+n\log{k})\log{n})$ time and  $O(\M +n\log{k})$ space.
The two algorithms can be combined in a crossing-oblivious
scheme within the same bounds. 
We show how to apply the RIC framework to handle non-standard characteristics of generalized
Voronoi diagrams, including sites (and bisectors) of non-constant complexity, 
sites that are not enclosed in their Voronoi regions, empty Voronoi
regions, and finally, disconnected bisectors and disconnected Voronoi regions. 
The Hausdorff Voronoi  diagram finds direct applications in VLSI~CAD.
\end{abstract}

\section{Introduction} 
\label{sec:intro}
The Voronoi diagram is a powerful geometric partitioning structure
that finds diverse applications in science and engineering~\cite{Aurenbook}.
In this paper we consider the 
\emph{Hausdorff Voronoi diagram} of \emph{clusters of
points} in the plane, a generalization of Voronoi diagrams based on the
Hausdorff distance function, which 
has applications 
in  predicting (and evaluating) faults in VLSI layouts 
and other geometric networks embedded in the plane.

Given a family $\AbsFamily$ of $k$ clusters of points in the plane,
where the total number of points is $n$ ($n=|\cup\AbsFamily|$),  the
\emph{Hausdorff Voronoi diagram} of $\AbsFamily$  is
a  subdivision of the plane into maximal regions such that 
all points within one region have the same \emph{nearest cluster} (see Figure~\ref{fig:hvd}a).
The distance between a point $t\in \mathbb{R}^2$ and a cluster $P\in
\AbsFamily$ is measured by 
their Hausdorff distance, 
which equals the \emph{farthest distance} between $t$ and $P$, $\df{t, P} = \max_{p \in
  P}d(t,p)$ 
and $d(\cdot,\cdot)$ denotes
the Euclidean distance\footnote{Other metrics, such as the
$L_p$ metric,  are possible.} between two points in the plane.
No two clusters in $\AbsFamily$  share a common point.

Informally, the Hausdorff Voronoi diagram is a \emph{min-max} type of
diagram.
The opposite \emph{max-min} type,  
where distance is minimum
and the diagram is farthest,
 is also of interest,
see e.g.,~\cite{fpvd2011cg,fcvd,hks1993dcg}.
Recently both  types of  diagrams  
have been combined to determine 
\emph{stabbing circles}  for sets of line segments in the plane~\cite{CKPSS17}.
We remark that the Hausdorff diagram 
is different in nature from the 
\emph{clustering
  induced Voronoi diagram} by Chen\etal~\cite{chen2013civd}, where sites can be all subsets
of points 
and the \emph{influence
  function} reflects a collective effect of all points in a
site.

The Hausdorff Voronoi diagram finds direct applications in Very Large Scale Integration
(VLSI) circuit design. It can be used to model the location of defects
falling over parts of a network that have been embedded in the plane,
destroying its connectivity.
It has been used extensively by the semiconductor industry
to  estimate the \emph{critical area} of a 
VLSI layout for 
various types of \emph {open faults}, see e.g.,
\cite{P11,CAA}. Critical area is a measure reflecting the
sensitivity of a design to random manufacturing defects.
The diagram can find
applications in geometric networks embedded in the plane, such as transportation networks,
where critical area may need to be extracted for the purpose of
flow control and disaster avoidance.

\paragraph{Previous work.}
The Hausdorff Voronoi diagram was first considered by Edelsbrunner~{et~al.}~\cite{EGS1989}
under the name \emph{cluster Voronoi diagram}.
The authors showed that its  combinatorial
complexity is $O(n^2 \alpha(n))$,  and 
gave a divide and conquer construction algorithm of the same time
complexity,
where $\alpha(n)$ is the inverse Ackermann function.
These bounds were later improved to $O(n^2)$ by Papadopoulou and
Lee~\cite{PL04}.
When the convex hulls of the clusters are disjoint~\cite{EGS1989} or
\emph{non-crossing} (see Definition~\ref{def:noncrossing})~\cite{PL04},
the combinatorial complexity of the  
diagram is $O(n)$.
The $O(n^2)$-time algorithm of Edelsbrunner\etal
is optimal in the worst case. It
exploits the equivalence of the Hausdorff diagram
to the upper envelope of a family of $k$ lower envelopes (one for each
cluster) in an arrangement of
planes in $\mathbb{R}^3$. 
However, it remains quadratic even if the 
diagram has complexity $O(n)$. To continue our description, we need the following.

\begin{definition}
\label{def:noncrossing}
Two clusters $P$ and $Q$ are called \emph{non-crossing},  
if the convex hull of $P\cup Q$ admits at most two supporting line segments with one endpoint in $P$
and one endpoint in $Q$. 
If  the convex hull of $P\cup Q$ admits more than two such supporting segments, then $P$
and $Q$ are called 
\emph{crossing} (see  Figure~\ref{fig:hvd}b).
\end{definition}

The combinatorial complexity of the Hausdorff Voronoi diagram
is $O(n+m)$, where $m$ is the number of \emph{crossings} between pairs
of crossing clusters (see
Definition~\ref{def:crossing}), and this is tight~\cite{ep2004algorithmica}. 
The number of crossings $m$ is upper-bounded by
the number of
supporting segments between pairs of crossing clusters.
In the worst case, $m$ is $\Theta(n^2)$.
Computing the Hausdorff Voronoi diagram in subquadratic time 
when $m$ is $O(n)$ (even if $m=0$)
has not been an easy task.
For  non-crossing clusters ($m=0$), the Hausdorff Voronoi diagram 
is an instance of \emph{abstract Voronoi diagrams}~\cite{klein-habilitation}. 
But a bisector can have complexity
$\Theta(n)$, thus, if we directly apply the randomized incremental 
construction for abstract Voronoi diagrams~\cite{klein-avd}
we get an $O(n^2\log n)$-time algorithm, and this 
is not easy to overcome (see~\cite{CKLP-alg}).
When clusters are crossing, their bisectors
are disconnected curves~\cite{PL04}, and thus, they 
do not satisfy the basic axioms of abstract Voronoi diagrams.

For non-crossing clusters, Dehne~et~al.~\cite{Dehne-coarse} gave the first
subquadratic-time algorithm to compute the Hausdorff diagram, in time
$O(n\log^5 n)$ and space $O(n\log^2 n)$.\footnote{The time complexity claimed in Dehne et al. is $O(n\log^4{n})$, 
however, in reality the described algorithm requires
$O(n\log^5{n})$ time~\cite{anil-communication}.}
Recently, Cheilaris\etal presented a randomized incremental
construction for this problem, which is 
based on \emph{point location} in a 
\emph{hierarchical dynamic
data structure}~\cite{CKLP-alg}.
The expected running time of this algorithm is 
$O(n\log{n}\log{\nclus})$ and the expected space complexity is
$O(n)$~\cite{CKLP-alg}.
However, this approach 
does not easily generalize to crossing clusters.

\begin{figure}
 \centering
\includegraphics{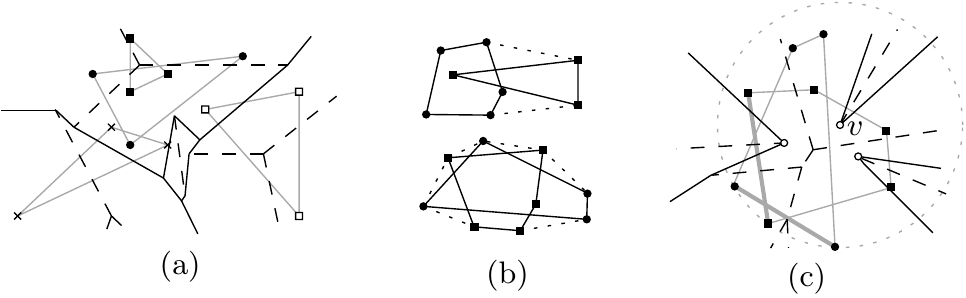}
\caption{(a) The Hausdorff Voronoi diagram of a family of four clusters, where each cluster contains three points. 
(b) A pair of clusters. Above: clusters are non-crossing;  below: 
clusters  are crossing. (c) 
Two crossing clusters $P$ and $Q$ (filled disks and squares, resp.), and 
their Hausdorff Voronoi diagram (black lines). The region of $P$ is disconnected into three faces.  
A crossing mixed vertex $v$, the circle passing through the three points that induce $v$ (dotted lines), and two diagonals
 of $P$ and $Q$ related to $v$ (bold, grey). 
}
\label{fig:hvd}
\end{figure}

\paragraph{Our Contribution.}
In this paper  we revisit  the randomized incremental construction
 for the
Hausdorff diagram and obtain three new results, which complete our
investigation on randomized incremental construction algorithms for
this diagram. The obtained results are especially relevant to 
the setting driven by our application, where
crossings may be present 
whose number is expected to be small
(typically, $\M=O(n)$).
We follow the randomized
incremental construction (RIC) framework introduced by
\mbox{Clarkson\etal~\cite{Clarkson_rand_sampling_2,CMS93}}.
We show how to
efficiently apply this framework to construct  a generalized Voronoi
diagram in the presence of several non-standard features:
(1) bisectors between pairs of sites can each have complexity  $\Theta(n)$; (2) sites need not
be  enclosed in their Voronoi regions; (3) Voronoi regions can be
empty;  and (4) bisector curves may  be disconnected. 
Note that a direct application of the framework would yield an $O(n^2\log n)$ 
(or $O(kn\log n)$) -time algorithm, even for a diagram of complexity $O(n)$.

First, we consider non-crossing clusters, for which the complexity of the
diagram is $O(n)$.
Our algorithm runs  in expected 
$O(n\log{n}+k\log{n\log{k}})$ time  
and deterministic $O(n)$ space. 
In comparison to our previous algorithm~\cite{CKLP-alg},
the construction is considerably simpler and it 
slightly improves its  
time complexity. 
We give the construction for both a \emph{conflict} and a
\emph{history graph}, where the latter 
is an \emph{on-line} variation of the algorithm (see Section~\ref{subsec:ncr-hist}). 

Then, we  consider arbitrary clusters of points. Allowing clusters to cross adds an entire new challenge to the
construction algorithm: 
a bisector  between two clusters consists of multiple connected components,  
thus, one Voronoi region may disconnect in several faces. 
We show how to overcome this challenge 
on a conflict graph and
derive an algorithm  whose expected time and space requirements 
are respectively  
$O(\M \log{n}+n\log{k}\log{n})$ and  $O(\M +n\log{k})$.
To the best of our knowledge, this is the first time the RIC
framework is applied to the construction of 
a Voronoi diagram with disconnected bisectors and disconnected regions.

Finally, we address the question of
deciding which algorithm to use on a given input, without a prior
knowledge on the existence of crossings. 
Deciding whether the input clusters are crossing (or not)
may require quadratic time  by itself,
because the convex hulls of the input clusters may have
quadratic number of intersections, even if the clusters are
non-crossing. In Section~\ref{sec:oblivious}, we show how to only
detect crossings
that are relevant to the
construction of the Hausdorff Voronoi diagram, and thus, provide a crossing-oblivious algorithm
that combines our two previous algorithms, while keeping intact the time
complexity bounds.

\section{Preliminaries}
\label{sec:prelim}

Let $\AbsFamily$ be a family of  $\nclus$
clusters of points in the plane, and let $n$ be the total number of points in $F$; 
no two clusters  
share a point.  
We assume that each cluster 
equals the vertices on its convex hull, as only points on a convex hull
may have non-empty 
regions in the Hausdorff Voronoi  diagram. For simplicity of presentation, we follow a general position
assumption  that  no four points  lie on the same circle.
This general position assumption can be removed similarly
to an ordinary Voronoi diagram of points, e.g., following techniques of symbolic perturbation~\cite{seidel1998nature}. 

The \emph{farthest Voronoi diagram} of a cluster $\aClus$, for brevity $\FVD{\aClus}$, 
is a partitioning of the plane into regions 
such that the 
\emph{farthest Voronoi region} of a point $c\in \aClus$~is 
\[
\freg{\aClus}{\pointFromClus}  = \{\aPoint \mid
\forall c' \in \aClus \setminus\{\pointFromClus\} \colon
d(\aPoint,\pointFromClus) > d(\aPoint,c') \}. \]
Let $\fskel{C}$ denote the graph structure of
$\FVD{\aClus}$, 
$\fskel{C} =  \Plane
\setminus \bigcup_{c \in C} \freg{C}{c}$. 
If  $|C|>1$, $\fskel{C}$ is well known to
be a tree; we  assume that $\fskel{C}$ is rooted at a point at infinity 
on an arbitrary unbounded Voronoi edge.
If $C = \{c\}$, let $\fskel{C} = c$.

The \emph{Hausdorff Voronoi diagram}, for brevity $\HVD{\AbsFamily}$,
is a partitioning of the plane into 
regions 
such that the  \emph{Hausdorff Voronoi region} of 
a cluster $C \in \AbsFamily$ is   
\[\hreg{\AbsFamily}{C}  = \{p \mid \forall C'\in \AbsFamily\setminus\{C\} \colon \df{p,C}
< \df{p,C'} \}.\] 
The region $\hreg{\AbsFamily}{C}$ is further subdivided
into subregions by the $\FVD{C}$.
In particular, the \emph{Hausdorff Voronoi region} of a point $c \in C$ is 
\[\hreg{\AbsFamily}{c}  = \hreg{\AbsFamily}{C} \cap \freg{C}{c}.\]

Figure~\ref{fig:hvd}a illustrates the Hausdorff Voronoi diagram of
a family of four clusters, where 
the convex hulls of the clusters are  
shown in grey lines.
Solid black lines indicate 
the \emph{Hausdorff Voronoi edges} bounding the 
regions of individual clusters, and  the
dashed lines indicate the finer subdivision, which is induced by the farthest Voronoi diagram of \mbox{each
cluster.} Figure~\ref{fig:hvd}c illustrates the diagram for a different faily of clusters following the same drawing conventions. 

The Hausdorff Voronoi edges are portions of \emph{Hausdorff bisectors}
between  pairs of clusters.
The Hausdorff bisector of  two clusters $P,Q \in F$ is  $\bh{P,Q} = \{y \mid  \df{y,P} = \df{y,Q}
\}$; see the solid black lines in Figure~\ref{fig:hvd}c. 
It is a subgraph of $\fskel{P\cup Q}$, and 
it consists of one (if $P,Q$
are non-crossing) or more (if $P,Q$ are crossing) unbounded polygonal
chains~\cite{PL04}. In Figure~\ref{fig:hvd}c the Hausdorff bisector of
the two clusters has three such chains. Each   vertex 
of $\bh{P,Q}$ is the center of a circle passing through two points of
one cluster and one point of another, which entirely  encloses $P$ and
$Q$, see, e.g.,  the dotted circle  
with center at vertex $v$ in Figure~\ref{fig:hvd}c.

\begin{definition}
\label{def:crossing}
A  vertex on the bisector $\bh{C,P}$, induced by two points $c_i, c_j \in C$ and a point $p_l \in P$,  
is called \emph{crossing}, if there is a diagonal $p_lp_r$ of $P$
that crosses the diagonal $c_ic_j$ of $C$, and  all points $c_i,c_j,p_l,p_r$ are on 
the convex hull of $C \cup P$. (See  
vertex $v$ in Figure~\ref{fig:hvd}c.)
The total number of crossing vertices along the bisectors of all pairs of clusters is the \emph{number of crossings}
and this is denoted by $\M$.
\end{definition}

The Hausdorff Voronoi diagram contains three types of
vertices~\cite{ep2004algorithmica} (see Figures~\ref{fig:hvd}a and c):
  (1)~\emph{pure} Voronoi vertices, 
equidistant to three clusters; 
  (2)~\emph{mixed} Voronoi vertices,
  equidistant to three points of two clusters; and 
  (3)~\emph{farthest}   Voronoi vertices,
equidistant to
  three points of one cluster. 
The mixed vertices, which are induced by two points of cluster $C$ (and one point of another cluster), are
  called $C$-\emph{mixed} vertices, and  they are incident to edges of 
  $\FVD{C}$. 
The Hausdorff Voronoi edges are polygonal lines (portions of Hausdorff
bisectors) 
 that connect pure Voronoi vertices.
Mixed Voronoi
 vertices are vertices of Hausdorff bisectors. 
They are characterized as crossing or non-crossing according to Definition~\ref{def:crossing}.
The following property is crucial for our algorithms.  

\begin{lemma}[\cite{ep2004algorithmica}] 
\label{prop:con-comp}   
Each face of a (non-empty) region $\hreg{F}{C}$ 
intersects $\fskel{C}$ in one non-empty
connected component.
The intersection points delimiting this component are $C$-mixed vertices. 
\end{lemma}

 Unless stated otherwise, we use a refinement of the Hausdorff Voronoi diagram as derived by the 
 \emph{visibility decomposition} of each
 region   $\hreg{\AbsFamily}{p}$~\cite{PL04} (see Figure~\ref{fig:hvd-insertion}a):
 for each vertex $v$ on 
the boundary of $\hreg{\AbsFamily}{p}$ draw
the line segment $p{v}$, as restricted within $\hreg{\AbsFamily}{p}$.
Each face $\f$ within  $\hreg{F}{p}$ is convex. 
In Figure~\ref{fig:hvd-insertion}a, the edges of the visibility
decomposition in $\hreg{F}{p}$ are shown in bold.

\begin{observation}
\label{obs:vb-face-def}
A face $\f$ of $\hreg{F}{p}$, $p \in P$, borders the regions of $O(1)$
(at most 3) other
clusters  
in $F \setminus \{P\}$. 
\end{observation}

Given a cluster $P$ and its farthest  Voronoi diagram $\FVD{P}$, our
algorithms often need to answer the following query, 
termed the \emph{segment query}:

\begin{definition}[Segment Query~\cite{CKLP-alg}]\label{def:segm-query}
Given two clusters $P, C \in F$,  $\FVD{P}$, and a line segment $uv
\subset \fskel{C}$ such that $\df{u,C} < \df{u,P}$ and $\df{v,C} >
\df{v,P}$,
find the point $x \in uv$ that is equidistant to  
$C$ and $P$.
\end{definition}

Following~\cite{CKLP-alg}, this 
query can be answered  in $O(\log{|P|})$ time using the
\emph{centroid decomposition} of $\FVD{P}$, see~\cite{CKLP-alg} and references therein. 
The centroid decomposition of $\FVD{P}$ is obtained 
by  recursively breaking $\fskel{P}$ into subtrees of its
\emph{centroid}, 
where the centroid of a tree with $h$ vertices
is a vertex whose removal decomposes the tree into
subtrees with at most $h/2$ vertices each.

\paragraph{Overview of the RIC
  framework~\cite{Clarkson_rand_sampling_2,CMS93}.}
\label{sec:ric}
The framework of the randomized incremental construction (RIC)
to compute a Voronoi diagram,  
inserts sites (also called objects) one by one, in random order, each time recomputing the target diagram. 
The diagram is viewed as a collection of
 \emph{ranges} (also called regions\footnote{To avoid confusion with
   Voronoi regions we use the term \emph{ranges} in this paper.}),  \emph{defined} and without \emph{conflicts} with respect to the set of sites inserted so far. 
 To update the diagram efficiently after each insertion, 
 a \emph{conflict} or a
 \emph{history graph}  is maintained.   
An important prerequisite for using the framework is that \emph{each range must be defined by a constant number  of objects}.

The \emph{conflict graph} is a bipartite graph, 
where one group of nodes corresponds to the ranges of the diagram
  defined by the sites inserted so far,  and the other group 
  corresponds to sites that have not yet been inserted. 
In the conflict graph, a range and a site are connected by an arc if
and only if they are in conflict. 
A RIC algorithm using a conflict graph 
is efficient if the following \emph{update condition} is satisfied at each incremental step: 
(1) Updating the set of ranges defined and without conflicts over the current subset of objects 
requires time proportional to the number of ranges deleted 
or created during this step; and
(2) Updating the conflict graph  
requires time proportional to the number of arcs of the conflict graph
that are added or removed
during this step. 

The expected 
time and space 
complexity for a RIC algorithm is as stated in~\cite[Theorem 5.2.3]{BY1998book}:
Let  $f_0(r)$ be  the expected number of ranges in the target diagram of a random sample of $r$ objects, and let $k$ 
be the number of insertion steps of the algorithm. 
Then:
(1)~ The expected
number of ranges created during the algorithm 
is $O\left(\sum\limits_{r=1}^k(f_0(r)/r)\right)$. 
(2)~If the update condition holds, the expected 
time and space of
the algorithm is 
$O\left(k\sum\limits_{r=1}^k(f_0(r)/r^2)\right)$.

The \emph{history graph} 
is a directed acyclic graph defined as follows. At each step of the algorithm,  
the nodes of the history graph correspond to 
all ranges created by the algorithm so far. 
The nodes with zero out-degree, termed \emph{leaf nodes},
correspond to the ranges present in the current target diagram. 
The nodes that have outgoing edges, termed \emph{intermediate nodes}, correspond to 
the ranges that have been deleted already. 
Intermediate nodes  are connected, by their outgoing edges, to 
the nodes whose insertion  caused the deletion of their ranges. 
The latter nodes are referred to as the \emph{children} of the former nodes. 
The \emph{update condition for a history graph} is the following:
(i) {The out-degree of each node is bounded by a constant; and (ii) the existence of a conflict 
between a given range and a given object can be tested in constant time.} 
The bounds on the time complexity of a randomized incremental algorithm
using a history graph  are the same as those given in item (2) above;
the storage required by the algorithm equals the number of ranges created during the 
algorithm and its expected value is bounded as in item (1) above. 

The rest of the paper is organized as follows. In Section~\ref{sec:non-cr} we give an algorithm to construct $\HVD{F}$ for the input cluster families where all clusters are pairwise non-crossing. In Section~\ref{sec:cr} we give an algorithm for the inputs where clusters may be crossing. 
Both algorithms follow the RIC framework (with the conflict graph) reviewed in the above paragraphs, therefore each of the two sections consists of: (1) the definition of objects, ranges and conflicts, (2) the procedure to insert an object and to update the conflict graph, and (3) the analysis of time and space complexity (note that the analysis differs substantially from simply applying ready theorems about RIC). In addition, Section~\ref{sec:non-cr} contains an adaptation of the algorithm to work with history graph instead of the conflict graph; this adaptation in turn follows the above presentation scheme of three items (1)-(3). Finally, in Section~\ref{sec:oblivious} we show how to combine the algorithm of Section~\ref{sec:non-cr} and the one of Section~\ref{sec:cr} in an efficient algorithm for the case where it is not known whether the input clusters have crossings.

\section{Constructing $\HVD{F}$ for non-crossing clusters}
\label{sec:non-cr} 
Let the clusters in  the input family $F$ be pairwise non-crossing.
Then each Voronoi region is connected 
and the combinatorial complexity of the Hausdorff Voronoi diagram is $O(n)$.

\begin{figure}
\begin{minipage}{0.29\linewidth}
\includegraphics{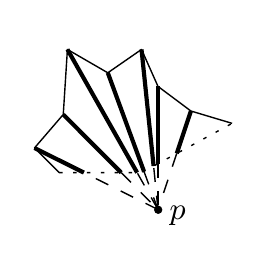}
\\
\centering
(a)
\end{minipage}
\begin{minipage}{0.33\linewidth}
\includegraphics[page=1]
{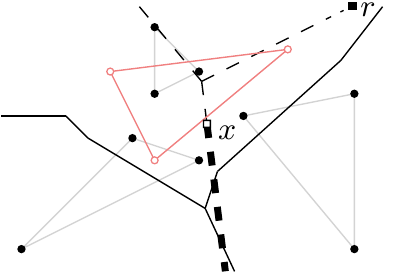}
\\
\centering
(b)
\end{minipage}
\hfill
\begin{minipage}{0.33\linewidth}
\includegraphics[page=2]
{hvd-ins-after}
\\
\centering
(c)
\end{minipage}
\caption{ Left: (a)  Visibility decomposition of the diagram. Right: Insertion of a cluster $C$ (unfilled disks): (b) $\fskel{C}$ (dashed) rooted at $r$, its 
active subtree (bold) rooted at $x$.  
(c) After the 
insertion: $x$ is a $C$-mixed vertex of the HVD.}  
\label{fig:hvd-insertion}
\end{figure}

Let $S \subset F$ such that $\HVD{S}$ has been computed.
Let $C \in F \setminus S$; our goal is to insert $C$ and obtain $\HVD{S \cup
  \{C\}}$. 
We first introduce the following definition for an \emph{active
  subtree} of $\fskel{C}$ and its \emph{root}.
\begin{definition}
\label{def:active}
 Traverse $\fskel{C}$, starting at its root, and 
 let $x$ be the first point we encounter  in the closure of 
$\hreg{S \cup\{C\}}{C}$. 
We refer to the subtree of  $\fskel{C}$ rooted at $x$ as the
\emph{active subtree} of $\fskel{C}$ and denote it by
$\act{S}{C}$. Let $x$ be the root of $\act{S}{C}$.
\end{definition}

Note that the roots of $\fskel{C}$ and $\act{S}{C}$ may coincide, in
which case,  $\act{S}{C}= \fskel{C}$.
In Figure~\ref{fig:hvd-insertion}b, $\act{S}{C}$ is shown in bold dashed
lines superimposed on $\HVD{S}$.
Figure~\ref{fig:hvd-insertion}c illustrates $\HVD{S \cup \{C\}}$.
By Lemma~\ref{prop:con-comp}, and  
since the Voronoi
regions of non-crossing clusters are connected, we have the following property.

\begin{property}
The region $\hreg{S \cup\{C\}}{C}\neq\emptyset$
if and only if $\act{S}{C}\neq\emptyset$.
The root of  $\act{S}{C}$ is a $C$-mixed vertex of
$\HVD{S \cup\{C\}}$ (unless $\act{S}{C}= \fskel{C}$).  
\end{property}

\deleted{
By Lemma~\ref{prop:con-comp}, and  
since the Voronoi
regions of non-crossing clusters are connected, we have the following:
$\hreg{S \cup\{C\}}{C}\neq\emptyset$
if and only if $\act{S}{C}\neq\emptyset$.
Further, the root of  $\act{S}{C}$ is a $C$-mixed vertex of
$\HVD{S \cup\{C\}}$, unless $\act{S}{C}= \fskel{C}$.
}

\subsection{Objects, ranges and conflicts}
\label{sec:ins-ncr} 
We formulate the problem 
of computing $\HVD{\absFamily}$ 
in terms of \emph{objects}, \emph{ranges} and \emph{conflicts}, see Section~\ref{sec:prelim}.
The objects are clusters in $F$.  
The ranges are the refined faces of 
$\HVD{S}$, as refined by the visibility decomposition of $\HVD{S}$
(see Section~\ref{sec:prelim} and Figure~\ref{fig:hvd}c). 
 A range corresponding to a face $\f$, where $\f \subset \hreg{S}{p}$
 and  $p\in P$,  
is said to be \emph{defined by} the cluster $P$ and by the remaining clusters
in $S$ whose Voronoi regions border $\f$.
By Observation~\ref{obs:vb-face-def}, 
there are at most three such clusters;
thus, 
the RIC framework is applicable.
Point $p$ is called the \emph{owner} of range $f$ and  $\hreg{S}{p}$.

\begin{definition}[Conflict for non-crossing clusters]
\label{def:confl} A range $\f$ is \emph{in conflict} with a cluster $C \in F \setminus S$, if 
$\act{S}{C}$ is not empty and its root $x$ lies in $\f$. 
A \emph{conflict} is a triple $(\f,x,C)$;
the \emph{list of conflicts} of range $\f$ is denoted by $\conf{\f}$.
\end{definition}

The following property is implied by  
Lemma~\ref{prop:con-comp}. 
It 
is essential for our algorithm.

\begin{lemma}
\label{cor:one-conf}
Each cluster 
in $F \setminus S$  has at most one conflict with the ranges in
$\HVD{S}$.
If a cluster $\curClus \in F \setminus S$ has no conflicts, 
then $\hreg{S \cup \{C\}}{C} = \emptyset$,  thus, $\hreg{F}{\curClus} = \emptyset$. 
\end{lemma}

In the following section, we present the variant of our
algorithm, i.e.,  the procedure to insert a cluster, that is based on a \emph{conflict graph}. 
In Section~\ref{subsec:ncr-hist}, we present  an adaptation
using  a \emph{history graph}.

\subsection{Insertion of a cluster} 
\label{sec:ins-conf-ncr}

Suppose that $\HVD{S}$, $S\subset F$, and its  conflict graph with $F \setminus S$ have been
constructed.
Let $C \in F \setminus S$.
Using the conflict 
of $C$, we can easily compute $\hreg{S
  \cup\{C\}}{C}$. Starting at
the root of $\act{S}{C}$ (which is stored with the conflict), trace the
region boundary 
in an ordinary way
~\cite{PL04}.
The main problem that remains is to identify the conflicts for the new ranges of $\hreg{S
  \cup\{C\}}{C}$ and update the conflict graph.
We give the algorithm 
to perform these tasks in Figure~\ref{fig:algo-upd-ncr} 
as pseudocode and summarize it in the sequel.

To identify new conflicts we use the information 
stored with the ranges that get deleted.
Let $f$ be  a deleted range, and let $p$ be its owner ($\f \subset
 \hreg{S}{p}$). 
For each conflict  $(\f, y, Q)$ of $f$, where $Q$ is a cluster in $F \setminus S$ and $y$ is the root of $\act{S}{Q}$, we  compute
the new root of $\act{S \cup \{C\}}{Q}$, if different from $y$, and identify 
 the new range that 
contains  it.
To compute  the root of $\act{S \cup \{C\}}{Q}$, it is enough to traverse
$\act{S}{Q}$, searching for  an edge $uv$
that contains a point
equidistant from $Q$ and $C$
(see Line 12 in Figure~\ref{fig:algo-upd-ncr}). 
If such an edge $uv$ exists, we
perform a \emph{segment query} 
in $\FVD{C}$ (see Definition~\ref{def:segm-query}), to compute the point
equidistant from $C$ and $Q$ on $uv$ (a $Q$-mixed vertex)
If no such edge exists, then $\act{S \cup \{C\}}{Q}=\emptyset$ and no conflicts for $Q$ should be created.
The remaining algorithm is straightforward (see
Figure~\ref{fig:algo-upd-ncr}). Its correctness is shown in the following lemma. 

\begin{figure}
\centering
\includegraphics[width = 0.85\linewidth]{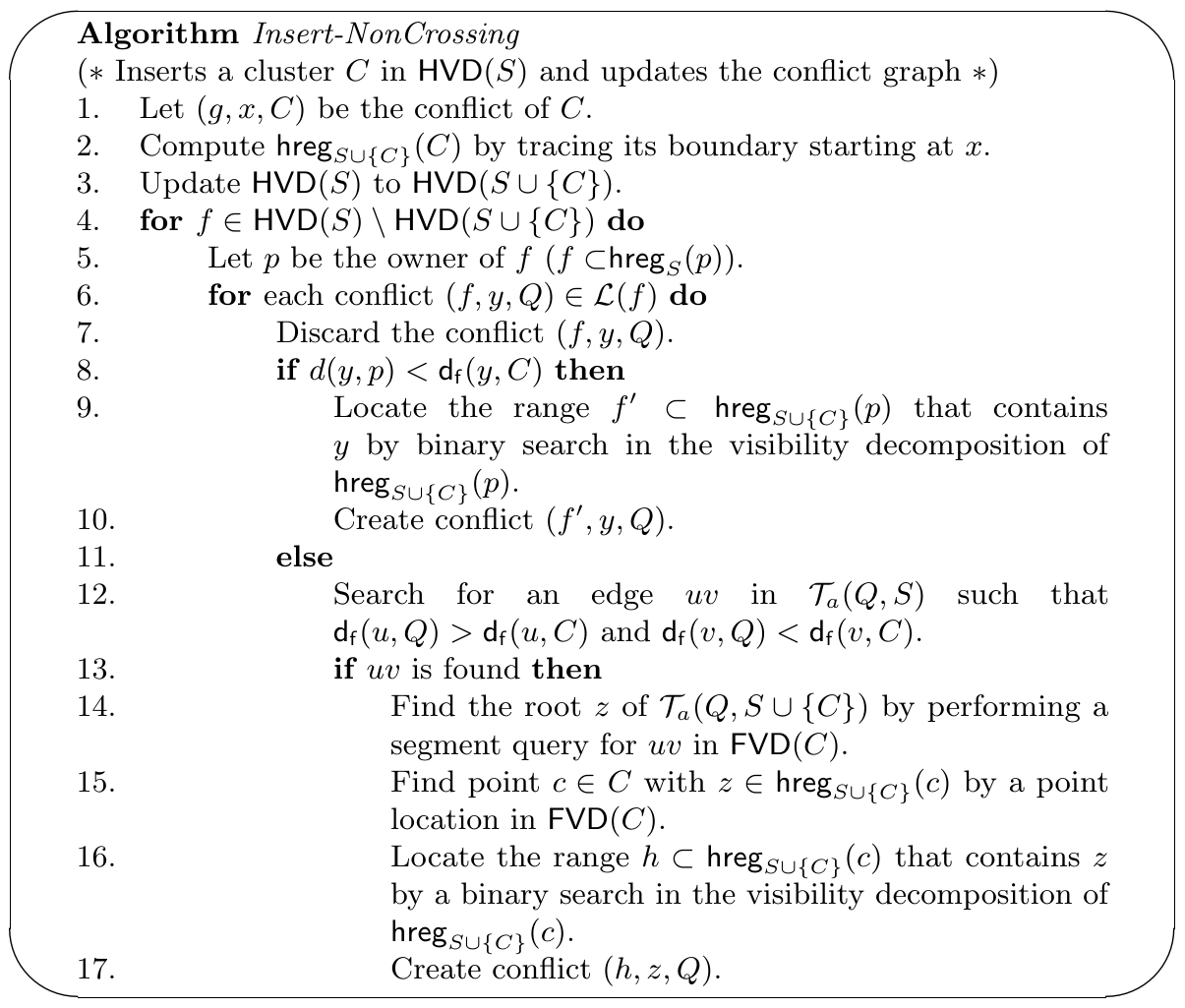}
    \caption{Algorithm to insert cluster $C$; case of non-crossing clusters.}
    \label{fig:algo-upd-ncr}
\end{figure}

\begin{lemma}
\label{lemma:non-cr-corr}
The algorithm \emph{Insert-NonCrossing} (Figure~\ref{fig:algo-upd-ncr}) is correct. 
\end{lemma}

\begin{proof}
The algorithm \emph{Insert-NonCrossing}, to insert a cluster $C$, processes the unique conflict $(g,x,C)$ of cluster $C$. 
Having point $x$ as a starting point for tracing the boundary of $\hreg{s \cup\{C\}}{C}$ (Line 2) is correct, since 
$x$ lies on that boundary as implied by Definition~\ref{def:confl}.
The Clarkson-Shor framework~\cite{Clarkson_rand_sampling_2,CMS93} (see also Section~\ref{sec:ric}) implies that processing the conflicts of every deleted 
range (the loop in Lines~4--16) is enough for repairing the conflict graph.
The only non-trivial parts of this loop
are: the condition in Line~8, and handling of its two possible outcomes respectively in Lines~9--10 and Lines 12--17.

In Line~8, it is enough to compare the distance from $y$ to only $p$ and $C$
because no other cluster may become the closest to $y$ as a result of
inserting $C$. Therefore by checking the condition in Line~8, we find out whether the owner of the face containing point $y$ 
in $\HVD{S \cup \{C\}}$ stays the same as in $\HVD{S}$, or $y$ actually belongs to the new region $\hreg{s \cup\{C\}}{C}$. 
The correctness of Lines~9--10, i.e, the case when  point $y$ stays in the region of $p$,  is easy to see. 
Suppose that $\df{y, C} < \df{y,Q} = \df{y,p}$, and that $\act{S \cup \{C\}}{Q} \neq \emptyset$ (Lines~12--17). 
In this case $y$ is no longer a part of the (updated) diagram $\HVD{S \cup \{C\}}$.
Clearly, $\act{S \cup \{C\}}{Q}$ is a subtree of  $\act{S}{Q}$, thus, there is exactly 
one edge $uv$ of $\act{S}{Q}$ such that 
$u\not\in\act{S \cup \{C\}}{Q}$ and $v\in\act{S \cup \{C\}}{Q}$. 
Edge $uv$ satisfies the condition of Line~12 by definition of an active subtree (see Definition~\ref{def:active}). 
The root of an active subtree cannot coincide 
with a vertex of $\fskel{C}$ due to the general position
assumption.
 Note that if we had no general position assumption, the root of an active subtree could  coincide with a vertex of $\fskel{C}$, but this  simple case  can be easily detected in Line~12 by checking whether the vertices of $\fskel{C}$ 
visited by the search are equidistant to $Q$ and to $C$.
 \qed
\end{proof}

The main result of this section is the following theorem, which we
prove in the remaining part of this section.

\begin{theorem}
  \label{th:time-compl}
The Hausdorff Voronoi diagram of $k$ non-crossing clusters
of total complexity $n$ can be computed in expected 
 $O(n\log{n} + k\log{n}\log{k})$ time and deterministic $O(n)$ 
space. 
\end{theorem}

Consider a random permutation $\{C_1,\dots,C_{\nclus}\}$ of the input family $\InpSet$,
and the sequence 
$\{\InpSet_0, \dots,\InpSet_{\nclus}\}$, 
where $\InpSet_0 = \emptyset$, $\InpSet_{i} = \InpSet_{i-1} \cup \{C_{i}\}$ and $\InpSet_{\nclus} = \InpSet$.

At step $i$ we insert cluster $C_i$ by performing the algorithm
\emph{Insert-NonCrossing} (Figure~\ref{fig:algo-upd-ncr}).
However, the update condition of the RIC framework 
does not hold,
thus, we cannot directly use it 
to obtain the time 
complexity of our algorithm.
\deleted{
The update condition of the RIC framework 
does not hold in our setting,
thus we cannot use the RIC framework directly to obtain the 
complexity of the algorithm.
}
To bound the expected total time required for updating the diagram, we
use the analysis of~\cite{CKLP-alg}, which applies to any randomized
incremental construction algorithm for the Hausdorff Voronoi diagram,
independently of the auxiliary data structure being used. 
The following lemma summarizes the result of~\cite{CKLP-alg}. 

\begin{lemma}[\cite{CKLP-alg}]
\label{lemma:updating-diagram}
During the course of  a randomized incremental construction 
of the Hausdorff Voronoi diagram of a family of non-crossing clusters, 
the total expected number of updates made to the diagram  is $O(n)$, and 
the total expected
time required for updating the diagram is $O(n \log n)$, where $n$ is the total 
number of points in all clusters.
\end{lemma}

We now bound the work to update the conflict graph.
We first give Lemmas~\ref{lemma:graph_size}--\ref{lemma:conf_size},
and then use them in Lemma~\ref{lemma:updating-cg} to derive this
bound.

\begin{lemma}
    \label{lemma:graph_size}
 The number of  
 arcs in the conflict graph, at any step, is $O(k)$.
 \end{lemma}

\begin{proof}
The statement is a direct implication of Lemma~\ref{cor:one-conf}. 
\end{proof}

\begin{lemma}
  \label{lemma:upd_cond}
  Updating the conflict graph at step $i$ 
  requires  
$O(\log{n}(N_i +   R_i))$ 
  time, where $N_i$ is the total number of edges dropped out of the active subtrees of clusters in  $F \setminus F_i$ 
  at step $i$, and   
 $R_i$ is the total number of conflicts deleted at
  step $i$.  
\end{lemma}

\begin{proof}
Updating the conflict graph corresponds to two nested for-loops
in Lines~4--16 of the algorithm in Figure~\ref{fig:algo-upd-ncr}.
Clearly the inner loop  (Lines~6--15) is performed $O(R_i)$ times in total. 
Inside this loop, a 
breadth-first search is performed 
that spends $O(\log{n})$ time per visited edge. 
By Lemma~\ref{cor:one-conf}, one active subtree is considered at most once during one step. 
All the visited edges, except the last one, are dropped out of the respective active subtree.  
It remains to show, that, except for the breath-first search, 
the rest of the work in any execution of the inner loop requires $O(\log{n})$ time. Indeed, 
it is a point location in Line~8, a segment query in  $\FVD{C_i}$ in Line~13, 
 a binary search 
in Line~9 or Line~15, and an insertion of a new conflict in  Line~10 or Line~17. 
For information on the segment query, see Section~\ref{sec:prelim}. 
\qed 
\end{proof}

\begin{lemma}
  \label{lemma:conf_size} 
  The expected total number of conflicts   deleted at step $i$
  of the randomized incremental algorithm is $O(k/i) + D_i$, where $D_i$ is the number of clusters in $F \setminus F_i$ 
  that used to have a conflict until step $i$, and do not have it any more. 
\end{lemma}

\begin{proof}
Each conflict deleted at step $i$ 
either induces a conflict with a (new) range in  
$\Ranges{F_{i}} \setminus \Ranges{F_{i-1}}$, 
or the corresponding cluster is counted by $D_i$. Each cluster is counted at most once by $D_i$ 
due to Lemma~\ref{cor:one-conf}.  
Thus, the  total number of conflicts deleted at step $i$ equals   the total number 
 of conflicts of ranges inserted at step $i$ plus $D_i$. 
 We bound the expectation of the former number by backwards analysis. 
After step $i$ is performed, the number of conflicts in the conflict graph is $O(k)$ by Lemma~\ref{lemma:graph_size}. 
Fix one conflict of some cluster $C_k$, $k \geq i$; it is between $C_k$ and  a range of a cluster $C_j$, $j \leq i$; 
since the insertion order of clusters is random, the probability for $C_j$ to be inserted at step $i$ (i.e., to be $C_i$ in our notation) is $O(1/i)$. Summing this probability for all the conflicts, we obtain
that the expected number of conflicts inserted at step $i$ is  
 $O(k/i)$. \qed
\end{proof}

\begin{lemma}
\label{lemma:updating-cg}
The expectation of the total time required to update the conflict graph throughout the algorithm is  $O(n\log{n} + k\log{n}\log{k})$. 
\end{lemma}

\begin{proof}
Summing the bound of  Lemma~\ref{lemma:upd_cond} for all steps, we obtain that
 the   total time required to update  the conflict graph during all steps of the algorithm is 
  $\sum\limits_{i = 1}^k{O\left(\log{n}(N_i + R_i)\right)}$.
   An edge of $\fskel{C}$ of any cluster $C \in F$ 
  is dropped out from the active subtree of $C$ at most once. 
  The total number of edges in the farthest Voronoi diagrams
  of all clusters in $F$ is $O(n)$. 
  Thus the above sum is proportional to
  $O(\log{n})(n +\sum\limits_{i = 1}^k{R_i})$. 
  The expectation of this number is, by Lemma~\ref{lemma:conf_size},
  $O\left(\left(n+ \sum\limits_{i = 1}^k\left({k{/}i + D_i}\right)\right)\log{n}\right)$. Note that $\sum\limits_{i=1}^k{D_i} \leq k$, 
  since the  active subtree of a cluster can become empty at most once.  
  The claimed time complexity follows. \qed
\end{proof}

We remark that the total number of conflicts created throughout our algorithm is $O(n+k\log k)$ 
(as evident by the proof of Lemma~\ref{lemma:updating-cg}).

\begin{proof}[of Theorem~\ref{th:time-compl}]
The total expected time required for updating the diagram during all the steps of the algorithm is 
$O(n\log n)$ due to Lemma~\ref{lemma:updating-diagram}.  
The total expected time required for updating the conflict graph is  $O(n\log{n} + k\log{n}\log{k})$ due to 
Lemma~\ref{lemma:updating-cg}.

  The space requirement at any step is proportional to 
  the combinatorial complexity of the Hausdorff Voronoi diagram, which is $O(n)$, plus the 
 total number of 
 arcs of the conflict graph at this step, which is at most $k$ by Lemma~\ref{lemma:graph_size}. 
 Hence the claimed $O(n)$ bound holds. \qed 
 \end{proof}

\subsection{Adapting the algorithm of Section~\ref{sec:non-cr} to using a history graph}
\label{subsec:ncr-hist}

Let $\Hist{\InpSet_i}$ denote the history graph
at step $i$ of the incremental algorithm.
  $\Hist{\InpSet_0}$ is a single node that corresponds
 to the whole $\mathbb{R}^2$.  
  For $i \in \{1,\dots,{\nclus}\}$, $\Hist{\InpSet_i}$ 
consists of all nodes and
  arcs of $\Hist{\InpSet_{i-1}}$, and in addition it contains the following:
      (i) A node for each new range in $\Ranges{\InpSet_{i}} \setminus \Ranges{\InpSet_{i-1}}$.
      These nodes are called the \emph{nodes of level $i$}.
     (ii) An arc connecting a deleted range  
      $\f \in \Ranges{\InpSet_{i-1}} \setminus \Ranges{\InpSet_{i}}$
      to 
every new range $\f' \in \Ranges{\InpSet_{i}} \setminus \Ranges{\InpSet_{i-1}}$ such that $\f'$ intersects $\f$. 
  
Suppose that $\HVD{\InpSet_{i-1}}$ and $\Hist{\InpSet_{i-1}}$ have
already been computed. 
To insert the next cluster $C_i$, we traverse $\Hist{\InpSet_{i-1}}$
from root to a leaf. 
Simultaneously, we 
traverse $\fskel{C_i}$, keeping track of the root $x$ of the 
active subtree $\act{\InpSet_j}{C_i}$ at  the current level $j$ of  
$\Hist{\InpSet_{i-1}}$.
When we reach a leaf of $\Hist{\InpSet_{i-1}}$, we trace the boundary of the Voronoi region 
$\hreg{\InpSet_i}{C_i}$, starting at the  root $x$ of $\act{\InpSet_i}{C_i}$, and 
update $\Hist{\InpSet_{i-1}}$ to become $\Hist{\InpSet_i}$.

In detail, the procedure at level $j$ is as follows:
Let $\f$ be the face in $\HVD{\InpSet_j}$ that contains the root $x$ of
$\act{\InpSet_i}{C_i}$. 
  Suppose that $\f$ is deleted at  step $\level$.
If {$\df{x, C_i} < \df{x,C_{\level}}$}, we search for the child $\f'$ of $\f$,
that has the same owner as $\f$, and contains $x$;
we move to the level $\level$, keeping $x$ intact,  and updating its face to be $
\f'$.
Else we search for the 
root $z$ of the (new) 
active subtree  
$\act{\InpSet_{\level} \cup \{C_i\}}{C_i}$ (the procedure to do this is the same as for the conflict graph, 
see Section~\ref{sec:ins-conf-ncr}). 
If $z$ is found,
we move to  level $\level$, replace $x$ by $z$ and the face $\f$ by $
\f' \subset \hreg{\InpSet_{\level}}{C_{\level}}$ that contains  $z$.
If $z$ is not found, 
the active subtree $\hreg{F_i}{C_i} = \emptyset.$

\begin{theorem}
\label{thm:hist-ncr}
The Hausdorff Voronoi diagram of $k$ non-crossing clusters
of total complexity $n$ can be computed by RIC with the history graph in expected 
 $O(n\log{n} + k\log{n}\log{k})$ time and expected $O(n)$ space. 
\end{theorem}

\begin{proof}
By Lemma~\ref{lemma:updating-diagram}, the total time required for updating the diagram during all steps of the algorithm 
is $O(n\log n)$.
Lemma~\ref{lemma:updating-diagram} also implies that the expected total number of ranges created during the algorithm is
 $O(n)$,  which
yields the bound on the expected storage requirement in the case of the
history graph.

Updating the history graph during step $i$ takes time
$O(\log{n}(N_i + K_i))$, where $N_i$ is the number 
of edges of $\fskel{C_i}$ that do not belong to the
active subtree $\act{F_i}{C_i}$, and thus, they  are
eliminated by the breath-first search.
$K_i$ is the number of clusters in the sequence $\{C_1,\dots, C_{i-1}\}$
  that change the root of the active subtree 
  as we move in the history graph level by level.
 By the backwards analysis, the expectation of $K_i$ is $O(\log{i})$.
 Summing over all $k$ steps gives us $O(k\log{k})$. The total expected
 running time  of the algorithm 
using the history graph is thus $O(n\log{n} + \nclus\log{n}\log{\nclus})$. \qed
\end{proof}

\section{Computing $\HVD{F}$ for arbitrary clusters of points}
\label{sec:cr}
In this section we drop the assumption that clusters  in the input family $F$
are pairwise non-crossing. 
This raises a major difficulty that 
Hausdorff bisectors may consist
of more than one polygonal curve and Voronoi regions may be
disconnected.
The definition of conflict 
from  Section~\ref{sec:ins-ncr}
no longer guarantees a correct diagram.
We thus need a new conflict definition. 
For a region $r \subset \mathbb{R}^2$, its boundary and its closure are denoted, respectively, $\partial r$ and $\cl{r}$. 

Let $C_1,\dots, C_k$ be a random permutation of clusters in $F$. 
We incrementally compute  $\HVD{F_i}$, $i = 1,\ldots,k$, where $F_i = \{C_1,C_2, \dots, C_i\}$. 
 At each step $i$, 
cluster $C_i$ is inserted in $\HVD{F_{i-1}}$. 
 We maintain the conflict graph (see Section~\ref{sec:ric}) between
 the ranges of $\HVD{F_{i}}$ and the clusters in $F \setminus
 F_{i}$. 
Like in Section~\ref{sec:non-cr}, ranges correspond to faces of 
$\HVD{F_{i}}$ as partitioned by the visibility decomposition.

When inserting cluster $C_i$, we compute $\hreg{F_i}{C_i}$, where $F_i = F_{i-1}\cup \{C_i\}$, 
using the information provided by the conflicts of 
$C_i$ with the ranges of $\HVD{F_{i-1}}$. 
From these conflicts we must be able to find  
at least one point in each face of $\hreg{F_i}{C_i}$.
For this purpose it is sufficient (by Lemma~\ref{prop:con-comp}) that
at every step $i$ of the algorithm,
we maintain
information on the $Q$-mixed vertices of $\HVD{F_i \cup \{Q\}}$, for every 
$Q \in F\setminus F_i$.
However, this is not sufficient to apply the Clarkson-Shor technique, since  
we need the ability to determine  new conflicts from
the conflicts of ranges that get deleted.
Due to this requirement, it is essential  that 
$Q$ is in conflict  not only  with
the ranges  of $\HVD{F_i}$  that contain  $Q$-mixed
vertices in $\HVD{F_i \cup \{Q\}}$, 
as Lemma~\ref{prop:con-comp} suggests, but also 
with all the ranges
that intersect the boundary of $\hreg{F_i \cup \{Q\}}{Q}$.

Let $\f$ be a range of $\HVD{F_{i}}$,  $\f \subset \hreg{F_{i}}{p}$,
where $p\in P$ and $P\in F_i$. Let $Q$ be a cluster in   $F \setminus F_{i}$.   
We define a conflict between  \mbox{$f$ and $Q$} as follows.

\begin{definition}[Conflict for arbitrary clusters]
  \label{def:confl-arb} 
Range $\f$, $\f \subset \hreg{F_{i}}{p}$, is \emph{in conflict} with cluster $Q$,
if $\f$  intersects  the boundary of $\hreg{F_{i}\cup\{Q\}}{Q}$. 
The \emph{vertex list} $V(\f, Q)$ of this  conflict 
 is 
the list of all vertices  and all endpoints of 
$\bh{P,Q} \cap \cl{\f}$, ordered  in clockwise angular order
around $p$.
\end{definition}
 
The following observation is essential to efficiently update the
conflict graph.

\begin{observation}
\label{obs:point-bisector}
 $\bh{P,Q} \cap \cl{\f} = \bh{p,Q} \cap \cl{\f}$. 
 The order of vertices in the list $V(\f, Q)$ coincides with the natural order
 of vertices along $\bh{p,Q}$.
\end{observation}

Note that
$\bh{p,Q}$ is a single
 convex chain~\cite{PL04}, unlike
$\bh{P,Q}$, and this is why it is much  simpler to process the former
chain rather  than the latter.
Figure~\ref{fig:branches} shows bisector $\bh{P,Q}$ and a range $f$ intersected by it. 
The boundary $\partial\f$ of this range consists of four parts: the top side 
is a portion of a pure edge of $\HVD{F_i}$; the
bottom chain, shown in bold, is  a portion of $\fskel{P}$; 
the two sides 
are edges of the visibility
decomposition of $\hreg{F_{i}}{p}$.

 \begin{figure}
\begin{minipage}{0.49\linewidth}
\centering
\includegraphics[page=1]{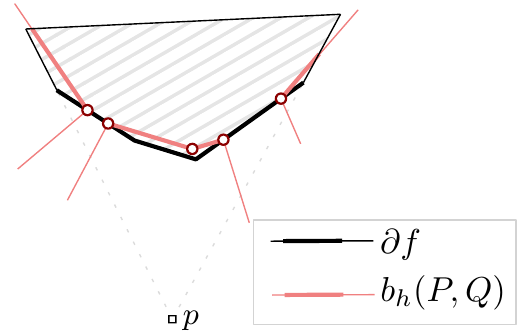}
\caption{A range $\f \subset \hreg{F_{i}}{p}$ (shaded);
Bisector $\bh{P,Q}$, 
$Q \in F \setminus F_{i}$;  
vertices of $\bh{P,Q}$ (unfilled circle marks).}
\label{fig:branches}
\end{minipage}
\hfill
\begin{minipage}{0.49\linewidth}
\includegraphics{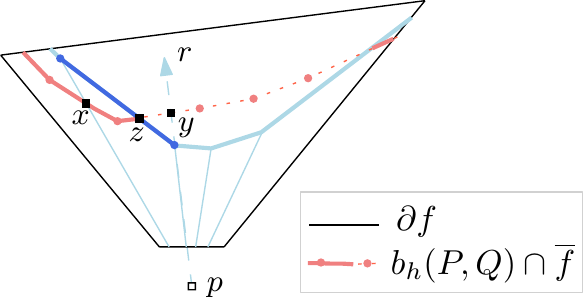}
\caption{An old range $\f \subset \hreg{F_{i-1}}{p}$;  
New ranges of $\hreg{F_i}{p}$ derived from $\f$ (bounded by light 
solid lines). 
$\bh{P,Q}  \cap \overline{\f}$
where $Q$ is  a cluster in conflict with $\f$.  
} 
\label{fig:new-ranges}
\end{minipage}
\end{figure}

We proceed with a procedure to insert a cluser and to update the conflict graph after this insertion (see Section~\ref{sec:ins-conf-cr}), 
and afterwards we analyze 
the complexity of this procedure and of the whole RIC algorithm (see Section~\ref{sec:compl-analysis-cr}). 

\subsection{Insertion of a cluster}
\label{sec:ins-conf-cr}

\paragraph{Insert $C_i$ into $\HVD{F_{i-1}}$.} 
We compute all faces of $\hreg{F_i}{C_i}$  
 by tracing their boundary,  starting at the vertices in the vertex lists of the conflicts of $C_i$.
In particular, while there are unprocessed vertices in these vertex
lists, pick one such vertex $v$ and trace the connected component of
$\partial \hreg{F_i}{C_i}$ adjacent to~$v$. All vertices on that
connected component are in
the vertex lists of conflicts of $C_i$; mark them as processed. 

The insertion of $\hreg{F_i}{C_i}$
results in deleting some
ranges, which we call \emph{old ranges}, and inserting some other ones, which we call  
 \emph{new ranges}.
We have two types of new ranges:
\emph{type (1)}: the ranges in  $\hreg{F_i}{C_i}$; 
        and  \emph{type (2):} 
the new ranges in the Hausdorff regions of clusters in $F_{i-1}$ derived from the old ranges, as a result of the  
insertion of $C_i$. The type (2) ranges are derived from the deleted ranges
of clusters in $F_{i-1}$ (see Figure~\ref{fig:new-ranges}).

\paragraph{Update the conflict graph.}
For each cluster $Q \in F\setminus F_i$  
in conflict with
at least one deleted range, 
    compute the conflicts 
of $Q$ with the new ranges. We compute the conflicts with ranges of type \mbox{(1) and (2)} separately as follows.

\paragraph{Ranges of type (1).}

Consider the ranges in  
$\hreg{F_i}{C_i}$. 
We follow the bisector $\bh{Q, C_i}$ within
$\hreg{F_i}{C_i}$, while computing this bisector on the fly. 
For each face $f \subset \hreg{F_i}{C_i}$ that is encountered as we walk on
$\bh{Q, C_i}$,
we also discover vertices in $V(f,Q)$.
Since $\bh{Q, C_i}\cap f$  may consist of
several components (see e.g. bold lines in Figure~\ref{fig:branches}),  
$f$ can be encountered a number of times;
each time, we augment  $V(f,Q)$
independently. 
We do 
this by inserting in $V(f,Q)$
the vertices on the branch of 
$\bh{Q, C_i}\cap f$ that had just been encountered, including its
endpoints  on $\partial f$.
The position in $V(f,Q)$ of the insertion can be determined 
by binary search.

\paragraph{Ranges of type (2).}

Let $f \subset \hreg{F_{i-1}}{p}$, $p \in P$, be a deleted range that
had been 
in conflict with $Q$.
Recall that the vertex list $V(f,Q)$ corresponds to $\bh{P,Q} \cap \cl f$, which coincides with  $\bh{p,Q} \cap \cl f$ (see Defintion~\ref{def:confl-arb} 
and  Observation~\ref{obs:point-bisector}). 
In the following we always operate with the latter one.
 For each new range $f'$ such that  $f' \subset f$ 
 and $f'$ is in conflict with $Q$, we need to compute  list $V(f',Q)$. 
 Observe, that the union of $\bh{p,Q} \cap \cl{f'}$, for all such ranges $f'$,
is $(\bh{p,Q} \cap \cl{f}) \setminus \hreg{F_i}{C_i}$.\footnote{Note that some of the new ranges of type (2) in fact consist of portions of two or more distinct old ranges. However, here we treat each such range as a group of ranges, as subdivided by old ranges. After the vertex lists of this group are found, it is easy to merge these ranges into a single range, and their vertex lists into a single list.}
We now introduce some notation. 
 
We call the maximal contiguous portions of $\bh{p,Q} \cap \cl{f}$, 
outside $\hreg{F_i}{C_i}$, the
\emph{active parts} of $\bh{p,Q} \cap \cl{f}$. 
The \emph{non-active parts} of $\bh{p,Q} \cap \cl{f}$  are its maximal
contiguous portions inside $\hreg{F_i}{C_i}$.
Figure~\ref{fig:new-ranges} 
shows the active and the non-active  parts of $\bh{p,Q}  \cap \overline{\f}$ 
by (red) bold and 
dotted lines respectively.
Note that one active (resp., non-active) part may consist of multiple
polygonal curves.  
A point incident to one active and one non-active part is called a
\emph{transition point},   
see e.g. point $z$ in Figure~\ref{fig:new-ranges}. 

Transition points 
lie in $\bh{Q,C_i} \cap \partial \hreg{F_i}{C_i}$; they are
used as starting points to compute
conflicts for ranges of type~(1).
Our task is to determine all  active parts of $\bh{p,Q} \cap
\cl{f}$, their incident transition points, and to create conflicts induced by these active parts.

We process  active and non-active parts of $\bh{p,Q} \cap \cl f$ sequentially: 
\begin{itemize}
\item  
For a non-active part,
we trace it in $\hreg{F_i}{C_i}$ 
in order
to determine the 
transition point where the next active part begins.
\item
For an active part, 
we process sequentially the new ranges of $\HVD{F_i}$ that are intersected by it. 
For each such range $f' \subset f$, we 
compute $V(f',Q)$, given the point $x$ where the active part enters
$f'$.
In particular, we find 
the point $z$ where it exits $f'$; after that the 
list $V(f',Q)$ 
can be easily derived from the portion of  
$\bh{p,Q}$ between $x$ and $z$. A procedure 
to find point $z$ is detailed in Lemma~\ref{lemma:xy}. 
Point $z$ is the endpoint of the active part that
we were processing, thus, $z$ is a transition point. 
 
\end{itemize}

\begin{lemma}
\label{lemma:xy}
Point $z$ can be determined in $O(\log{n})$ time. 
\end{lemma}
\begin{proof}

To find point $z$,  consider the rightmost 
ray $r$ originating at $p$ and passing through $\partial f'$.
If $\bh{p,Q} \cap \cl \f$ intersects $r$ (see Figure~\ref{fig:new-ranges}), let $y$ be the point of this intersection. 
Otherwise,
we let $t$ be the rightmost endpoint of  $\bh{p,Q} \cap \cl \f$ to the left of $r$. 
If  $t \in \fskel{P}$, set $z=t$, otherwise set $y=t$.

If $y \in \partial f'$, then we set $z=y$. In this case the active part of $\bh{p,Q} \cap \cl \f$ 
enters the next new range $f''$ at point $y$. 
In Figure~\ref{fig:new-ranges}, this case is illustrated by point $x$ that 
plays the role of $z=y$.  

If  $y$ lies outside $\cl{f'}$ ($y \not\in \cl{f'}$),  
we determine point $z$ as 
the unique point on $\bh{p,Q}$, such that $z$ is between $x$ and $y$,
 and    $z \in \partial f'$. See Figure~\ref{fig:new-ranges}. 

Suppose $y$ lies outside $\cl{f'}$. 
In particular, by construction of $y$, it must lie in  $\cl{\f}\setminus \cl{f'}$.
The portion of $\bh{P,Q}$ between $x$ and $y$ lies entirely in $\cl{f}$, $x \in \cl{f'}$, and $y \not\in \cl{f'}$. 
By the visibility properties of $\bh{P,Q}$ and $\bh{P,C_i}$,
the portion of  $\bh{P,Q}$ between $x$ and $y$ (excluding $x$) intersects $\partial f'$ exactly once, 
and the intersection point $z$ lies on the top side of $\partial f'$ (the portion of a pure edge of $\HVD{F_{i+1}}$ on $\partial f'$); 
see the dark blue line segment in Figure~\ref{fig:new-ranges} or the top edge of $\f$ in Figure~\ref{fig:branches}. 
Let $ab$ denote the top side of $\partial f'$, where $a$ precedes $b$ in the clockwise order around $p$.   
If $x \not\in ab$, 
the entire subsegment $az$ of $ab$ is closer 
 to $Q$ than to $p$, and the entire $zb$ is closer to $p$ than to $Q$. 
If $x \in ab$, the above property holds for $xz$ and $zb$.
Thus  $z$ can be determined by a segment query for $ab$ (resp.,  $xb$) in $\FVD{Q}$ (see Definition~\ref{def:segm-query} of the segment query). 
A segment query can be performed in $O(\log{n})$ time, see Section~\ref{sec:prelim}.

Points $y$ and $t$ can be found in $O(\log{n})$ time by a binary search in $V(f,Q)$. \qed
\end{proof}

The following lemma shows correctness of the algorithm; its time and space complexity is analyzed in Section~\ref{sec:compl-analysis-cr}.

\begin{lemma}
\label{lemma:cr-correctness}
The above 
algorithm 
correctly updates the Hausdorff Voronoi diagram and the conflict graph after insertion of $C_i$. 
\end{lemma}

\begin{proof}
Correctness of updating the Hausdorff Voronoi diagram follows from Lemma~\ref{prop:con-comp}. Indeed, each face of $\hreg{F_i}{C_i}$
is incident to at least one $C_i$-mixed vertex. Since all $C_i$-mixed vertices are in the vertex lists of the conflicts of $C_i$, 
all the faces of $\hreg{F_i}{C_i}$ are discovered. 

While updating the conflict graph, all the conflicts between new ranges and clusters in $F \setminus F_i$ are computed. 
Indeed, for each $Q \in F\setminus F_i$, 
while computing the ranges of type (2), the algorithm discovers all the edges of $\HVD{F_i \cup{Q}}$ outside $\hreg{F_i}{C_i}$.
Using the transition points 
found while computing ranges of type (2) as starting points, 
the procedure to compute ranges of type (1) determines all the edges of $\HVD{F_i \cup{Q}}$ inside $\hreg{F_i}{C_i}$.  
If it was not the case, then $Q$ would have a face in $\HVD{F_i \cup{Q}}$ 
bounded solely  by the edges induced by $C_i$ and $Q$, i.e., a face that lies inside the region of another cluster, which is not possible in the Hausdorff Voronoi diagram. \qed
\end{proof}

\subsection{Complexity analysis of the algorithm for arbitrary clusters}
\label{sec:compl-analysis-cr}

In this section we analyze the time and space complexity of the algorithm in Section~\ref{sec:ins-conf-cr}.
The conflicts are defined in a non-standard way, i.e., they have vertex lists whose complexity need not be constant. We use the Clarkson-Shor technique to bound the number of ranges created throughout  the algorithm (see Theorem 3), but we
cannot rely on the Clarkson-Shor technique to bound the total number of conflicts nor the time complexity of the algorithm,
because the update condition of the technique (see Section 2) is not satisfied. 
Our analysis can be seen as extending the Clarkson-Shor analysis to such a non-standard setting.

In the next two lemmas, we analyze respectively the time complexity of updating the conflict graph after insertion of  cluster $C_i$, 
and the total space required for the conflict graph at any step of the algorithm.  The expectation of the total number of conflicts created 
during the algorithm is bounded in Theorem~\ref{th:cr}  that states the overall result.

\begin{lemma}
\label{lemma:upd}
Updating the conflict graph after  
insertion of cluster  $C_i$ can be done in time 
$O((A(C_i) + L(C_i)+ V(C_i))\log{n})$, where 
$A(C_i)$ is the number of conflicts created and deleted (i.e., new and old conflicts), 
$L(C_i)$ is the total number of mixed vertices 
in  the  vertex lists of 
old conflicts that do not appear in the vertex lists of new conflicts,  
and $V(C_i)$ is the total 
size of the vertex lists of all  conflicts of all ranges of 
$\hreg{F_i}{C_i}$. 
\end{lemma}

\begin{proof}
First consider the time complexity of creating all the conflicts of the new ranges derived from $\hreg{F_i}{C_i}$ (i.e., all the conflicts of the new ranges of type (1)).  
This is the total time 
spent for this task for all clusters $Q \in F\setminus F_i$ that are 
in conflict with such new ranges. 

Tracing $\bh{Q,C_i}$ inside $\hreg{F_i}{C_i}$ requires time proportional to the complexity of
$\bh{Q,C_i} \cap \hreg{F_i}{C_i}$, times $O(\log{n})$. 
The complexity of $\bh{Q,C_i} \cap \hreg{F_i}{C_i}$ is the
 number of vertices of $\bh{Q,C_i} \cap  \hreg{F_i}{C_i}$
plus the number of times a boundary between two ranges of $\hreg{F_i}{C_i}$ was crossed by  
$\bh{Q,C_i}$. This equals 
the total size of the vertex
 lists of all conflicts between $Q$ and the ranges of $\hreg{F_i}{C_i}$. 
Summing up this number for each cluster $Q$, we obtain $O(V(C_i))$.    Thus the total 
time complexity of creating all the conflicts of the new ranges 
of type (1) is $O(V(C_i)\log{n})$. 

Now we analyze the  time complexity of creating all the conflicts with the new ranges of type (2), 
for all clusters $Q \in F \setminus F_i$ that were in conflict
with old ranges.   
The total time required for tracing the non-active parts of $\bh{p,Q} \in \overline\f$ for all $Q$
  is $O(L(C_i)\log{n})$  
similarly to the above, since
 all the conflicts corresponding to the non-active parts are discarded at this step.

What remains is to bound the time required for processing all 
the active parts of $\bh{p,Q} \in \overline\f$ for all pairs $f,Q$ such that $f$ is an old range, and 
 $Q$ was in conflict with $f$.
Note that the  vertices of  the active parts of $\bh{p,Q} \in \overline\f$
 are counted neither by $L(C_i)$ nor by $V(C_i)$.
To avoid 
visiting all vertices of the active parts and manipulating them explicitly, we 
do not store conflict lists together with conflicts, but rather we store them (united)
 with the owners of the corresponding ranges. 
In particular, at step   $i$, for each point $p$  such that $\hreg{F_i}{p} \neq \emptyset$, 
and each cluster $Q \in F\setminus F_i$ in conflict with some range of $\hreg{F_i}{p}$,  
we store the list of vertices and endpoints of $\bh{p,Q} \cap \hreg{F_i}{p}$, 
which is
the union of 
$\bh{p,Q} \cap \cl f$ for all ranges $f$ of $\hreg{F_i}{p}$. 
The conflict between $f$ and $Q$ then, instead of storing the list $V(f,Q)$, stores only 
the leftmost and the rightmost endpoints of $\bh{p,Q} \cap \cl f$.  
This way, the conflict between $Q$ and a new range $f'$ of type (2), $f' \subset \hreg{F_i}{p}$,
can be created in constant time, 
  once the two  endpoints of $\bh{p,Q} \cap \cl {f'}$ are known. 
Any binary search in the list of vertices and endpoints of $\bh{p,Q} \cap \cl f$, e.g., 
the search for point $y$, 
can be performed  in $O(\log{n})$ time.

To determine whether a given point $q \in f$ is in a new range of type (2) or not, it is enough to compare the distance from $q$ to the owner of $f$ and to cluster 
$C_i$. This is done by point location in $\FVD{C_i}$, and thus requires $O(\log{n})$ time. 

The total time required  to create the conflicts of type (2) 
and to find the starting points for conflicts of type (1) 
is therefore $O((A(C_i)+ L(C_i))\log{n})$. 

After all the new conflicts are created, for each deleted range $\f$, we delete all the conflicts of $\f$, which requires $O(A(C_i))$ time.

Finally, recall now that 
our procedure in fact creates first groups of smaller new ranges as subdivided by old ranges, and merges them in the true new ranges 
after their conflicts are computed. 
This requires additional $O(A(C_i)\log n)$ time, since the total number of additional smaller new ranges is linear in the number of  
borders between old ranges, that is, $O(A(C_i))$.
\qed
\end{proof}

\begin{lemma}
\label{lemma:lists}
At any step $i$, the total size of  vertex lists of all the conflicts is $O(n + \M + N_i)$, where 
$N_i$ is the number of conflicts in the current conflict graph.    
\end{lemma}

\begin{proof}
Consider a cluster $Q \in F \setminus F_i$. 
Members of the vertex lists of the conflicts of $Q$ are: (1) the 
 mixed vertices of $\HVD{F_i \cup \{Q\}}$ that 
bound the region of $Q$ in that diagram, $\hreg{F_i\cup{Q}}{Q}$,  and (2)
points that are not mixed vertices of  $\HVD{F_i \cup \{Q\}}$.
The total  number of points of latter type 
is proportional to the number of conflicts of $Q$. 
Indeed, there are at most two points of the latter type per one conflict of $Q$ with a range $\f$. 
Such points in $V(f,Q)$ are 
exactly the intersections between two
polygonal chains, $\bh{p,Q}$ and  $\bd \f \setminus \fskel{P}$, where $p \in P$ and $f \subset \hreg{F_i}{p}$. 
The first chain is concave (as seen from $p$), 
 and the second one is  convex. Thus    
they intersect in at most two points.  

Now we bound the number of mixed vertices  
in the vertex lists of all conflicts of $Q$.
Recall, that the mixed vertices in $V(f,Q)$ are $P$-mixed and $Q$-mixed vertices, where $P \in F_i$ is such that $\f \subset \hreg{F_i}{P}$. 
The total number of crossing mixed vertices in the vertex lists of all conflicts between $Q$
and ranges derived from the region of $P$ is linear in the number of crossings
between $P$ and $Q$. 
There are at most two non-crossing  $P$-mixed vertices in  $\vlist{\f,Q}$ for each such conflict between $\f$ and $Q$, 
since the boundary of $\f$ that is a portion of $\fskel{P}$ is connected. 
The total number of $Q$-mixed vertices in $\HVD{F_i \cup \{Q\}}$ is $O(|Q|)$.  
Thus the total number of mixed vertices in the vertex lists of all the conflicts of $Q$ 
is $O(|Q| + Cr(Q))$, where $Cr(Q)$ is the number of crossings between $Q$ and all the clusters in $F_i$. 

The claim follows by summing up the above quantities over all clusters $Q$ in $F \setminus F_i$. \qed
\end{proof}

We are now ready to state the main result of this section. 

\begin{theorem}
\label{th:cr}
The Hausdorff Voronoi diagram of a family $F$ of $k$ clusters of total complexity $n$ can be 
computed in $O((\M+n\log{k})\log{n})$ expected time and $O(\M+n\log{k})$ expected space.
\end{theorem}

\begin{proof} The expected space complexity of the algorithm is $O(n\log{k} + \M)$, implied by 
the fact that the Hausdorff Voronoi diagram of any subset of $F$ has complexity $O(n +\M)$~\cite{ep2004algorithmica}, and by  
Lemma~\ref{lemma:lists}, stating that the 
additional 
space required to store  vertex lists of the conflicts at each step is 
$O(n +\M)$.

We now bound the expected time complexity of the algorithm.  
To this aim we estimate the expectation of 
the total number of conflicts created during the course of the algorithm, and the 
expectation  of the sum of $L(C_i)$ and of  $V(C_i)$,  for $i=1,\dots,k$. 

To analyze the expected total number of conflicts created during the course of the algorithm, 
we need to estimate the expected 
number of ranges, i.e., faces, in $\HVD{R}$, where $R$ is a random $r$-sample of $F$.
The number of faces in a Hausdorff Voronoi diagram is proportional to the number of
 its mixed Voronoi vertices~\cite{ep2004algorithmica}. 
The number of non-crossing mixed vertices in $\HVD{R}$ is proportional to the total number of points in all clusters in $R$~\cite{ep2004algorithmica}.  
The expectation of the latter number is  $O(nr/k)$: 
for each of $n$ points in $F$, the probability that this point  appears in $R$ 
is the probability that its cluster appears in $R$, which is $r/k$.   
For a crossing mixed vertex $v$ 
induced by clusters $P,Q \in F$, the probability that $v$  appears in $\HVD{R}$ is at most
 the probability that both $P$ and $Q$ appear in $R$, which is $O(r^2/k^2)$. 
Summing over all crossings of clusters in $S$, we have that the expected number of crossing
 mixed vertices in $\HVD{R}$ is $O(\M r^2/k^2)$. Therefore, the expected number of ranges in $\HVD{R}$ is $O(nr/k + \M r^2/k^2)$.  
The Clarkson-Shor analysis (see Section~\ref{sec:prelim})
implies that the expected total number of  conflicts created during the course of the algorithm 
is $O(n\log{k} + \M)$. 

It remains to bound the expectation of the sum of $L(C_i)$ and of the sum of $V(C_i)$, 
 for $i=1,\dots,k$. 
We first note that since each deleted vertex was created at some point, the first sum is bounded by the second one. 
 Thus  we should bound the total number of
 vertices that appear in vertex lists of the conflicts during the course of the algorithm. 

Recall that
the total size of the vertex lists of the conflicts is proportional to the number of the conflicts plus the 
total number of the mixed vertices in these lists. The former quantity is expected $O(n\log{k} + \M)$ as shown above.
The total number of all possible crossing mixed vertices is  $O(\M)$, and 
 therefore we only need to bound the non-crossing mixed vertices.  
 
For a cluster $Q \in F$, there are $O(|Q|)$ non-crossing $Q$-mixed vertices at any step, 
at most one vertex per one  edge of $\fskel{Q}$. 
For an edge $e$ of $\fskel{Q}$, there can be up to $k-1$ possible $Q$-mixed vertices on $e$; 
each such vertex $v$ can be assigned weight equal to $\df{v,Q}$. If a vertex $v$ appears in the diagram at some step, 
any vertex $v'$ with greater weight is guaranteed not to appear. The sequence of insertions 
of the clusters is 
a random permutation of $F$, and it corresponds to a random permutation of the sequence of the weights of mixed vertices along $e$. 
The ones that appear on the diagram are exactly the ones that change minimum in the partial sequence so far, and the expected number of such minimum changes 
 is known to be $O(\log{k})$~\cite{DeBerg2008}.  
The overall number of edges in $\fskel{Q}$ for all $Q \in F$ is $O(n)$, thus the total number of non-crossing vertices 
that appear during the course of the algorithm is $O(n\log{k})$. 
Therefore the sum of $V(C_i)$,  for $i=1,\dots,k$, is $O(n\log{k} + \M)$.

The claim now follows from  Lemma~\ref{lemma:upd}. \qed 
\end{proof}

\section{A crossing-oblivious algorithm}
\label{sec:oblivious}
In this section we 
discuss how to 
compute the diagram $\HVD{F}$ if it is not known 
whether 
 clusters in the input family $F$ have crossings. 
Deciding fast whether $F$ has crossings
is not an easy task, because  the convex hulls of the 
clusters may have
a quadratic total number of intersections, even if the clusters are actually
non-crossing.

{We overcome this issue by combining the two algorithms while  staying
  within the best complexity bound (see Theorem~\ref{thm:oblivious}).}
 We start with the
algorithm of Section~\ref{sec:non-cr} and run it until we realize 
that the diagram cannot be updated correctly. 
If this happens, we terminate the algorithm of Section~\ref{sec:non-cr},
and run the 
algorithm of Section~\ref{sec:cr}. 
In particular, after the 
insertion of a cluster $C_i$ we perform a check. 
A positive answer to this check 
guarantees  that the region 
$\hreg{F_i}{C_i}$ is connected, and
thus, it has been computed correctly. A negative answer  indicates that 
$C_i$ has a crossing with some cluster which has already been inserted
in the diagram and its region, $\hreg{F_i}{C_i}$,    
may  be disconnected.
At the first negative check, 
we restart the computation of 
the diagram from scratch using the algorithm of Section~\ref{sec:cr}.
We can afford to run 
the latter algorithm, since
it is now certain that the input family of clusters has crossings.

The procedure of the check is based on the following property of the Hausdorff Voronoi diagram: 

\begin{lemma}[\cite{ep2004algorithmica}]
\label{prop:crossing-vs-not} 
If $\hreg{F_i}{C_i}$ is disconnected, 
then each  of its connected components is incident to a crossing $C_i$-mixed vertex of $\HVD{F_i}$.
\end{lemma}

The following  lemma provides the second ingredient of the check, that is, it shows how to 
 efficiently detect whether a given face of $\hreg{F_i}{C_i}$ has a crossing $C_i$-mixed vertex
on its boundary.

\begin{lemma}
\label{lemma:check}
For a connected component $f$ of $\hreg{F_i}{C_i}$, we can in time $O(|C_i|\log{n})$ 
detect whether there is a crossing $C_i$-mixed vertex on 
the boundary of $f$. 

\end{lemma}
\begin{proof}
Let $v$ be a $C_i$-mixed vertex on the boundary of $f$, 
and let $Q$ be the cluster, that together with $C_i$ induces the vertex $v$.
Vertex $v$ breaks $\fskel{C_i}$ in two (open) connected portions: the one that intersects $f$, 
and the one that does not intersect $f$; we call it $\mathcal{T}_v(C_i)$. 
Since the  former portion lies in $\hreg{F_i}{C_i}$, it clearly contains points that are closer to $C_i$ than to $Q$.

Our task is to check whether 
$\mathcal{T}_v(C_i)$ contains such points as well. 
If this is the case, then $C_i$ and $Q$ are crossing, and $v$ is a
crossing $C_i$-mixed vertex.
Otherwise, $v$ is a non-crossing $C_i$-mixed vertex.  This can be easily verified using Definition~\ref{def:crossing}.
Note that for any two different $C_i$-mixed vertices $u,v$, the portions $\mathcal{T}_u(C_i)$ and $\mathcal{T}_v(C_i)$
are disjoint. 
Thus, it remains to check  whether
$\mathcal{T}_v(C_i)$ contains points closer to $Q$ than to $C_i$,   
in time proportional to the number of edges in $\mathcal{T}_v(C_i)$.

Observe, that given two points  $u,v$ on an edge $e$ of  $\fskel{C_i}$, if both $u$ and $v$  
are closer to $Q$ than to $C_i$, 
then all the points on $e$ between $u$ and $v$ are also 
closer to $Q$ than to $C_i$. 
Indeed, consider the two closed  disks $D_u, D_v$ centered 
respectively at $u$ and at $v$ whose radii  equal $\df{u,C_i}$ and $\df{v,C_i}$. 
Since both $u,v$ are closer to $Q$ than to $C_i$, cluster $Q \subset D_u \cap D_v$. 
Further, the disk $D_w$ centered at any point $w \in e$ with radius $\df{w,C_i}$
 contains $D_u \cap D_v$, and thus it contains $Q$, which means that $w$ is closer to $Q$ than to $C_i$. 
 
 By the above observation, it is enough to check 
 the vertices 
of $\mathcal{T}_v(C_i)$, including the ones at infinity along the unbounded edges of $\mathcal{T}_v(C_i)$.
If all of them are closer to $Q$ than to $C_i$, then $v$ is a non-crossing mixed vertex. Otherwise, $v$ is crossing.
To check for one vertex we perform a point location in $\FVD{Q}$, which 
requires $O(\log{|Q|}) = O(\log{n})$ time.
\qed 
\end{proof}

We conclude with the following. 

\begin{theorem}
\label{thm:oblivious}
Let $F$ be a family of $k$ clusters of total complexity $n$.
There is an algorithm that computes $\HVD{F}$
as follows: if the clusters in $F$ are non-crossing,
the algorithm works in 
$O(n)$ space, and expected 
 $O(n\log{n} + k\log{n}\log{k})$ time.  
If the clusters in $F$ are crossing, 
the algorithm requires $O((\M+n\log{k})\log{n})$ expected time and $O(\M+n\log{k})$ expected space, where $m$ is 
the total number of crossings between pairs of clusters in $F$.
\end{theorem}

\begin{proof}
  We insert clusters of $F$ one by one in random order, 
  running the algorithm of Section~\ref{sec:non-cr}. 
 
After inserting a cluster $C_i$,
we run the procedure in the proof 
of Lemma~\ref{lemma:check} for the computed region of $C_i$. 
The total time spent by this procedure is $O(n\log{n})$, since it 
is performed at most once for each cluster and $\sum_{i=1}^k|C_i| = n$.
One of the following  cases occurs:

\begin{itemize}
\item[Case 1.]
After each insertion, the check returned a positive answer. By Lemma~\ref{prop:crossing-vs-not} each newly inserted region was connected, and therefore
all the connected components of each such region were inserted in the diagram. That is, 
 the computed  diagram is indeed $\HVD{F}$. In this case, the time complexity of the algorithm is  expected $O(n\log{n} + k\log{n}\log{k})$, and its space complexity is (deterministic) $O(n)$.  
\item[Case 2.] At some insertion the check returned a negative answer.
In that case, 
the algorithm of  Section~\ref{sec:non-cr} is aborted, and
the diagram is computed from scratch by the algorithm of Section~\ref{sec:cr}. 
The time and space requirements of 
this algorithm 
dominate 
the ones of the (partially performed) algorithm in
Section~\ref{sec:non-cr} and the algorithm of the check procedure. 
Thus the time and space complexity of the algorithm in this case equal
the ones 
of the algorithm of  Section~\ref{sec:non-cr}.\end{itemize} 

\paragraph{}
If the clusters in $F$ are pairwise non-crossing,
then the Hausdorff Voronoi regions are connected in the diagram of any subset of
$F$, and thus,
by Lemma~\ref{prop:crossing-vs-not}, the algorithm
follows Case 1.

If the algorithm follows Case~2, then by Lemma~\ref{prop:crossing-vs-not}, 
there 
is at least one crossing between clusters in $F$.
This
completes the proof.
\qed
\end{proof}

\bibliographystyle{splncs03}
{\small
\bibliography{hvd-bibl}
}

\end{document}